\documentclass[runningheads]{llncs}

\usepackage[T1]{fontenc}

\usepackage{enumerate}
\usepackage{algorithm}  
\usepackage{algorithmicx}
\usepackage{algpseudocode}  
\usepackage{amsmath}
\usepackage{amssymb}
\usepackage{booktabs}
\usepackage{appendix}

\newcommand{\lcm}{\mathrm{lcm}}
\newcommand{\setlcm}[2]{ \underset{#1}{\lcm}\left(#2\right) }
\newcommand{\bra}[1]{ \left\lvert #1\right\rangle}
\newcommand{\floor}[1]{\left \lfloor #1 \right \rfloor}
\newcommand{\ceil}[1]{\left \lceil #1 \right \rceil}
\newcommand{\Z}{\mathbb{Z} }
\newcommand{\set}[1]{\left \{ #1 \right \}}
\newcommand{\intset}[1]{\left [ #1 \right ]}
\newcommand{\sfrac}[1]{\frac{1}{\sqrt{#1}}}

\newcommand{\QFT}{\mathcal{QFT}}
\newcommand{\R}{\mathcal{R}}
\newcommand{\CNOT}{\mathcal{CNOT}}
\renewcommand{\H}{\mathcal{H}}
\newcommand{\U}{\mathcal{U}}
\newcommand{\C}{\mathcal{C}}

\newcommand{\QFIT}{\mathcal{QFT}^\dagger}
\newcommand{\pexp}[2]{e^{\frac{\imath2\pi#1}{#2}}}
\newcommand{\nexp}[2]{e^{-\frac{\imath2\pi#1}{#2}}}
\newcommand{\onarrow}[1]{\overset{#1}{\rightarrow}}

\makeatletter
\newenvironment{breakablealgorithm}
  {
   \begin{center}
     \refstepcounter{algorithm}
     \hrule height.8pt depth0pt \kern2pt
     \renewcommand{\caption}[2][\relax]{
       {##2\par}%
       \ifx\relax##1\relax 
         \addcontentsline{loa}{algorithm}{\protect\numberline{\thealgorithm}##2}%
       \else 
         \addcontentsline{loa}{algorithm}{\protect\numberline{\thealgorithm}##1}%
       \fi
       \kern2pt\hrule\kern2pt
     }
  }{
     \kern2pt\hrule\relax
   \end{center}
  }
\makeatother

\renewenvironment{proof}[1][Proof]{\vspace*{12pt}\noindent{\bf #1: }}{$\square\,$\newline}

\begin{document}
\title{A Secure Multiparty Quantum Least Common Multiple Computation Protocol}
\titlerunning{A Secure Multi-party Quantum Least Common Multiple \dots}

\author{Zixian Li\inst{1} \and Wenjie Liu\inst{2}
}

\institute{School of Computer and Software, Nanjing University of Information Science and Technology, Nanjing, 210044, China\\ Email: \email{zixianli157@163.com}\inst{1},  \email{wenjiel@163.com}\inst{2}}

\maketitle           

\begin{abstract}
In this paper, we present a secure multiparty computation (SMC) protocol for least common multiple (LCM) based on Shor's quantum period-finding algorithm (QPA). Our protocol is based on the following principle: the connection of multiple periodic functions is also a periodic function whose period is exactly the least common multiple of all small periods. Since QPA is a probabilistic algorithm, we also propose a one-vote-down vote protocol based on the existing secure multi-party quantum summation protocol, which is used to verify the results of the proposed LCM protocol. Security analysis shows that under the semi honest model, the proposed protocol
is secure with high probability, while the computational consumption remains at polynomial complexity. The protocol proposed in this paper solves the problem of efficient and secure multiparty computation of LCM, demonstrating quantum computation potential. 

\end{abstract}

\section{Introduction}

Secure multiparty computation (SMC) is a process of cooperative computation of $n$ parties $P_0$,$P_1$,$\cdots$,$P_{n-1}$. They have secret inputs $x_0$,$x_1$,$\cdots$,$x_{n-1}$ respectively at the beginning, and after computation, each party $P_i$ gets an output $f_i\left(x_0,x_1,\cdots,x_{n-1}\right)$ without getting any other secrets of other parties. SMC is a new technology that uses all parties' information to do compute without revealing anyone's privacy. The least common multiple (LCM) is an important function in elementary number theory, computer science, cryptography, etc. For example, in the addition of rational numbers, it is necessary to find the LCM of the denominators. However, the SMC for LCM is rarely studied, because compared with other computational problems, it's hard to simply decomposed the computation of LCM into existing cryptographic primitives. If safety is not considered, LCM can be calculated in polynomial complexity according to formula $\lcm(x,y)=\frac{xy}{\gcd(x,y)}$, where $\gcd(x,y)$ (greatest common divisor) can be calculated by the well-known Euclid algorithm. However, the formula is only applicable to two integers. Therefore, for $n>2$, the formula needs to be used to calculate LCM between each two parties, which increases the risk of information disclosure; On the other hand, calculating the numerator $xy$ will directly reveal the integer $x,y$ itself. In order to secure compute LCM, a simple idea is to decompose all private integers into prime factors, so as to transform the problem into the operation of the exponents of these prime factors\cite{2018Yang}. However, this idea is not efficient. Let the upper limit of the integer involved be $N=2^m$, and according to the prime number theorem, the number of prime numbers less than $N$ is $\mathcal{\pi}(N)\sim \frac{N}{\log N}$. In order to ensure the complete coverage of prime factors, the scheme should at least consider the exponents of $\mathcal{\pi}(N)=\Theta\left(\frac{N}{\log N}\right)=\Omega \left(\frac{N}{\sqrt{N}}\right)=\Omega \left(2^{\frac{m}{2}}\right)$ prime factors. That is, if all positive integers within $N$ are allowed to be input, the complexity is exponential. Or, we can assume that all prime factors belong to some complete set, which will lose the universality of the computation. Similarly, all methods based on prime factor decomposition cannot be efficient, and we need new ideas.

Quantum computers are considered to have more computing power than classical computers, although this has not been strictly proved. It is well known that quantum computing can achieve exponential acceleration for specific problems, such as Deutsch-Jozsa's algorithm\cite{1992Deutsch}, Simon's algorithm\cite{1997Simon}, and the most famous Shor's algorithm\cite{1994Shor,1997Shor}. Shor's algorithm is essentially a quantum period-finding algorithm (QPA)\cite{1994Shor,1997Shor,2000Nielsen}. It can obtain the period of any $m$-bit function in $O\left(\log m\right)$ oracle operations, which is impossible for any known classical algorithm. The special nature of QPA is exactly what we need. In fact, assume we have several functions with positive integer periods. If the outputs of these functions are connected together to form a new function, the period of the new function is exactly the LCM of all small periods. This allows us to use QPA to bypass the prime factorization and directly complete the LCM computation.

However, since QPA is probabilistic, the above method alone are not sufficient. Logically, we need a final voting process to check whether the answers output by QPA are indeed a common multiple of all people. The requirements are as follows: (1) ensure that each participant can vote yes or no; (2) only when everyone passes, will they finally pass; (3) don't divulge anyone's vote. This kind of voting is called \textbf{One-vote-down vote (OV)}. In 2016, Shi\cite{2016Shi} proposed a simple Secure multi-party quantum summation (SMQS) protocol, which is used to safely sum several integers. It is unconditionally secure, with polynomial communication cost and complexity. We transform Shi's SMQS protocol into an OV protocol as a subprogram of LCM computation protocol.

\noindent \textbf{Our contributions.} In this paper, we produce following contributions:

\begin{itemize}
\item We propose a Secure multiparty OV protocol based on Shi's SMQS protocol;

\item We propose a quantum SMC protocol for LCM based on QPA, by taking the proposed OV protocol as a subprogram;

\item We prove that the above protocols are correct and safe, and have polynomial complexity.
\end{itemize}

The following parts of this article are arranged as follows. In Section~\ref{sec2} we agree some symbols, give basic definitions, and briefly introduce QPA and Shi's SMQS protocol. In section~\ref{sec3}, we first propose a quantum OV protocol, and then propose a quantum LCM computation protocol by taking the former as a subprogram. We analyze the correctness, security and complexity of the proposed protocols in Section~\ref{sec4} and conclude in Section~\ref{sec5}.

\section{Preliminary}\label{sec2}
In this section we will do some preliminary. We agree some symbols in Section~\ref{sec2.1}, give basic definitions in Section~\ref{sec2.2}, and briefly introduce QPA and Shi's SMQS protocol in Section~\ref{sec2.3} and in Section~\ref{sec2.4} respectively.

\subsection{Symbol agreement}\label{sec2.1}
The meanings of some symbols used in this paper are shown in Table~\ref{table1}.

\begin{table}[h]
\centering
\caption{Some symbols and their meanings}\label{table1}
\begin{tabular}{c|c}
\toprule
\toprule
Symbols & Meanings\\
\midrule
\midrule
$\imath$   & Imaginary unit  \\
\midrule
$n$   & Number of participants   \\
\midrule
$m,N=2^m$  & Number of bits and upper limit of input value respectively\\
\midrule
$u,v$   & The number of input bits and output bits of function $f$ respectively \\
\midrule
$\intset{x}$  & Set $\set{0,1,\cdots,x-2,x-1}$ \\
\midrule
$x\lvert y$   & $x$ is a factor of $y$ \\
\midrule
$\setlcm{i\in A}{x_i} $  & For a index set $A$, calculate the least common multiple of all integers $x_i$ \\
\midrule
$\bra{\psi}_h$  & Quantum register $h$ is in state $\bra{\psi}$ \\
\midrule
$x\parallel y$   & For binary strings $x,y$,  connect them end to end to form a new string\\
\midrule
$(h,t)$   & For quantum registers $h,t$,  connect them end to end to form a new register\\
\midrule
$\floor{x},\ceil{x}$   & Round integer $x$ up and down respectively\\
\midrule
$f^{-1}(y)$   & For function $f:A\rightarrow B$, find the solution set $\set{ x\lvert f(x)=y,x\in A}$\\
\midrule
$\Z$   & The set of integers. \\
\bottomrule
\bottomrule
\end{tabular}
\end{table}

In addition, the unitary operators used in this paper are as follows (we use two $m$-qubits registers $h=\left(h_0,h_1,\cdots,h_{m-1}\right)$ and $t=\left(t_0,t_1,\cdots,t_{m-1}\right)$ for description).
\begin{enumerate}[(1)]
    \item {Hadamard operator $\H^{\otimes m}$: apply Hadamard gate $\H$ to each qubit.
    \begin{equation}
        \H^{\otimes m}:\bra{x}_h\rightarrow \sfrac{2^m}\sum_{j\in\intset{2^m} }{(-1)}^{x\cdot j}\bra{j}_h
    \end{equation}
    }
    \item {Copy operator $\CNOT^{\otimes m}$: for each qubit pair $h_i$ and $t_i$, apply controlled NOT gate $\CNOT$ to them, where $h_i$ is the control qubit and $t_i$ is the target qubit.
    \begin{equation}
        \CNOT^{\otimes m}:\bra{x}_h\bra{y}_t\rightarrow \bra{x}_h\bra{y\oplus x}_t
    \end{equation}
    }
    \item {Quantum Fourier transform $\QFT$ and its inverse transformation $\QFIT$.
    \begin{equation}
    \begin{aligned}
        &\QFT:\bra{x}_h\rightarrow \sfrac{2^m}\sum_{j\in\intset{2^m} }\pexp{xj}{2^m}\bra{j}_h\\
        &\QFIT:\bra{x}_h\rightarrow \sfrac{2^m}\sum_{j\in\intset{2^m} }\nexp{xj}{2^m}\bra{j}_h
    \end{aligned}
    \end{equation}
    }
    \item {\cite{2016Shi} Phase operator $\U_+$, power operator $\C_j$ and modular multiplication operator $\U_{\times q}$ ($q$ must be an odd integer).
    \begin{equation}
    \begin{aligned}
        &\U_+:\bra{x}_t\rightarrow \pexp{x}{2^m}\bra{x}_t \\
        &\C_j:\bra{j}_h\bra{x}_t\rightarrow \bra{j}_h\U_+^j\bra{x}_t= \pexp{xj}{2^m}\bra{j}_h\bra{x}_t\\
        &\U_{\times q}:\bra{j}_h\rightarrow \bra{jq\mod 2^m}_h
    \end{aligned}
    \end{equation}
    We abbreviate $\bra{jq\mod{2^m}}_h$ as $\bra{jq}_h$.
    }
\end{enumerate}

\subsection{Problem definition}\label{sec2.2}

\begin{definition}[Least common multiple (LCM)] For any positive integer $x_0,x_1,\cdots,x_{n-1}$, the least common multiple $\setlcm{i\in \intset{n}}{x_i}$ is the minimum positive integer y that satisfies $\forall i\in \intset{n}, x_i\lvert y$.
\end{definition}

\begin{definition}[Least common multiple problem] Let $n$ parties $P_0,P_1,\cdots$, $P_{n-1}$ have secret positive integers $x_0,x_1,\cdots,x_{n-1}\in \intset{2^m}$ respectively. The least common multiple problem is to compute $y=\setlcm{k\in \intset{n}}{x_k}$, while any $P_i$ cannot get any privacy of other parties.
\end{definition}

\begin{definition}[One-vote-down vote (OV)] Let $n$ parties $P_0,P_1,\cdots,P_{n-1}$ have secret Boolean numbers $c_0,c_1,\cdots,c_{n-1}\in \set{0,1}$ respectively. The one-vote-down vote is to compute $y=\prod_{k\in \intset{n}}{c_k}$, while any $P_i$ cannot get any privacy of other parties.
\end{definition}

\begin{remark} In an OV protocol, only when everyone passes, will they finally pass. Obviously, it is equivalent to the logical multiplication of $n$ Boolean numbers $c_i$, i.e., only when $c_0=c_1=\cdots=c_{n-1}=1$, the result is 1, otherwise it is $0$.
\end{remark}

\subsection{Quantum period-finding algorithm}\label{sec2.3}

Quantum period-finding algorithm is described in \textbf{Algorithm~1}, with reference to Nielsen\cite{2000Nielsen}.

\begin{breakablealgorithm}
\caption{\raggedright\textbf{Algorithm~1} Quantum period-finding algorithm (QPA)}
\begin{algorithmic}[1]
\renewcommand{\algorithmicrequire}{\textbf{Input}}
\renewcommand{\algorithmicensure}{\textbf{Output}}
\Require A function $f:\intset{2^u}\rightarrow \intset{2^v}$ with a positive integer period $T<2^v$, where for each pair of $j\ne j'\in \intset{2^u}$, $f(j)=f(j')$ only if $j\equiv j'(\mod T)$;
\Ensure $T$;
\renewcommand{\algorithmicrequire}{\textbf{Success probability}}
\Require $O\left( \frac{1}{\log{\log{T}} }\right)$;
\State Prepare two quantum registers $h,t$ of $u,v$ qubits respectively, initialized as $\bra{0}_h\bra{0}_t$;
\State Apply $\H^{\otimes u}$ on $h$:
\begin{equation}\bra{0}_h\bra{0}_t\onarrow{\H^{\otimes u}}\sfrac{2^u}\sum_{j\in \intset{2^u}}\bra{j}_h\bra{0}_t;\end{equation}
\State Apply $\U_f:\bra{j}_h\bra{0}_t\rightarrow \bra{j}_h\bra{f(j)}_t$ on $h,t$: 
\begin{equation}
\begin{aligned}
&\sfrac{2^u}\sum_{j\in \intset{2^u}}\bra{j}_h\bra{0}_t \onarrow{\U_f}\sfrac{2^u}\sum_{j\in \intset{2^u}}\bra{j}_h\bra{f(j)}_t \\
&=\sfrac{T}\sum_{l\in \intset{T}}\sfrac{2^u}\sum_{j\in \intset{2^u}}\pexp{jl}{T}\bra{l}_h\bra{\widehat{f}(l)}_t,
\end{aligned}
\end{equation}
where $\bra{\widehat{f}(l)}=\sfrac{T}\sum_{k\in \intset{T}}\nexp{lk}{T}\bra{f(k)}$;
\State Apply $\QFIT$ on $h$: 
\begin{equation}
\sfrac{T}\sum_{l\in \intset{T}}\sfrac{2^u}\sum_{j\in \intset{2^u}}\pexp{jl}{T}\bra{l}_h\bra{\widehat{f}(j)}_t \onarrow{\QFIT}\sfrac{T}\sum_{l\in \intset{T}}\bra{\phi}_h\bra{\widehat{f}(l)}_t,
\end{equation}
where $\phi\approx 2^u\frac{l}{T} $;
\State Measure $h$: 
\begin{equation}
\sfrac{T}\sum_{l\in \intset{T}}\bra{\phi}_h\bra{\widehat{f}(l)}_t \onarrow{Measure}\bra{\phi}_h\bra{\widehat{f}(l)}_t,
\end{equation}
where $l\in \intset{T}$ is selected with equal probability $\frac{1}{T}$;
\State Use continued fraction expansion of $\phi$ to get $\frac{l_1}{T_1}=\frac{l}{T}$, where $\frac{l_1}{T_1}$ is the minimalist fraction of $\frac{l}{T}$. If $f(T)=f(0)$, output $T_1$; otherwise, repeat the above steps.
\end{algorithmic}
\end{breakablealgorithm}

In step 6, as mentioned by Shor, we should let $2^u\ge2\times {(2^v)}^2+1$, i.e., $u\ge2v+1=O(v)$ to ensure $\frac{l}{T}$ can be found by continued fraction expansion from $\phi$. On the other hand, if $l,T$ is coprime (with a probability $O\left( \frac{1}{\log{\log{T}} }\right)$), then we get $T_1=T$; otherwise, $T_1<T$, then the algorithm failed. Therefore, we should repeat the algorithm $O\left( {\log{\log{T}} }\right)\le O\left( {\log{\log{2^v}} }\right)=O\left( \log{v}\right)$ times to make sure the correct $T$ can be found. 

\subsection{Shi's secure multi-party quantum summation protocol}\label{sec2.4}

Shi's SMQS protocol is described in \textbf{Protocol~1}.

\begin{breakablealgorithm}
\caption{\raggedright\textbf{Protocol~1} Secure multi-party quantum summation (SMQS)\cite{2016Shi}}
\begin{algorithmic}[1]
\renewcommand{\algorithmicrequire}{\textbf{Input}}
\renewcommand{\algorithmicensure}{\textbf{Output}}
\Require $n$ parties $P_0,P_1,\cdots,P_{n-1}$ have secret integer $x_1,x_1,\cdots,x_{n-1} \in \intset{2^m}$ respectively;
\Ensure Each $P_i$ gets $y=\sum_{k\in \intset{n}}{x_k}\mod{2^m}$, without any privacy of other parties;
\State {For $P_0$, he  
\begin{enumerate}[(1)]
    \item prepares two $m$-qubit quantum registers $h,t$ initialized as $\bra{x_0}_h\bra{0}_t$;
    \item {applies $\QFT$ on $h$:
    \begin{equation}
    \bra{x_0}_h\bra{0}_t\onarrow{\QFT}\sfrac{2^m}\sum_{j\in \intset{2^m}}\pexp{x_0 j}{2^m}\bra{j}_h\bra{0}_t;
    \end{equation}}
    \item {applies $\CNOT^{\otimes m}$ on $h,t$, where $h$ controls $t$:
    \begin{equation}
    \sfrac{2^m}\sum_{j\in \intset{2^m}}\pexp{x_0 j}{2^m}\bra{j}_h\bra{0}_t\onarrow{\CNOT^{\otimes m}}\sfrac{2^m}\sum_{j\in \intset{2^m}}\pexp{x_0 j}{2^m}\bra{j}_h\bra{j}_t;
    \end{equation}}
    \item sends $t$ to $P_1$;
\end{enumerate}}
\State {For $P_i, 1\le i\le n-1$, he
\begin{enumerate}[(1)]
    \item prepares a quantum $m$-qubit registers $e_i$ initialized as $\bra{x_1}_{e_i}$,
    \item {applies $\C_j$ on $t,e_i$:
    \begin{equation}
    \begin{aligned}
    &\sfrac{2^m}\sum_{j\in \intset{2^m}}\pexp{\left(\sum_{k\in\intset{i}} {x_k} \right)j}{2^m}\bra{j}_h\bra{j}_t\bra{x_i}_{e_i}\\
    &\onarrow{\C_j}\sfrac{2^m}\sum_{j\in \intset{2^m}}\pexp{\left(\sum_{k\in\intset{i+1}} {x_k} \right) j}{2^m}\bra{j}_h\bra{j}_t\bra{x_i}_{e_i};
    \end{aligned}
    \end{equation}}
    \item sends $t$ to $P_{i+1\mod{n}}$;
\end{enumerate}}
\State {For $P_0$, he  
\begin{enumerate}[(1)]
    \item {applies $\CNOT^{\otimes m}$ on $h,t$, where $h$ controls  $t$:
    \begin{equation}
    \begin{aligned}
    &\sfrac{2^m}\sum_{j\in \intset{2^m}}\pexp{\left(\sum_{k\in\intset{n}} {x_k} \right) j}{2^m}\bra{j}_h\bra{j}_t\\
    &\onarrow{\CNOT^{\otimes m}}\sfrac{2^m}\sum_{j\in \intset{2^m}}\pexp{\left(\sum_{k\in\intset{n}} {x_k} \right) j}{2^m}\bra{j}_h\bra{0}_t;
    \end{aligned}
    \end{equation}}
    \item measures $t$. If $t$ is not $\bra{0}$, then rejects, otherwise continues;
    \item {applies $\QFIT$ on $h$:
    \begin{equation}
    \begin{aligned}
    \sfrac{2^m}\sum_{j\in \intset{2^m}}\pexp{\left(\sum_{k\in\intset{n}} {x_k} \right) j}{2^m}\bra{j}_h\onarrow{\QFIT}
    \bra{\sum_{k\in\intset{n}} {x_k} \mod 2^m}_h;
    \end{aligned}
    \end{equation}}
    \item measures $h$, and broadcasts $y=\sum_{k\in\intset{n}} {x_k} \mod 2^m$ to all other parties.
\end{enumerate}}
\end{algorithmic}
\end{breakablealgorithm}

Since operators $\QFT$ and $\C_j$ can both be divided into $O(m^2)$ controlled single quantum operators $\R_j=\begin{pmatrix}1  & 0\\ 0 & \pexp{}{2^j}\end{pmatrix}$ (in fact, $\U_+^{2^i}=\otimes_{k \in \intset{m}}{\R_{m-i-k}}$, and $\C_j=\U_+^j=\U_+^{\sum_{i\in \intset{m}}j_i2^i}=\prod_{i\in\intset{m}}{\left(\U_+^{2^i}\right)^{j_i}}$), and such two operators need to be executed $n$ times, so the total time and communication complexity of \textbf{Protocol~1} are $O(nm^2)$ and $O(nm)$ respectively.

\section{Proposed protocol}\label{sec3}

In this section we will present our protocol. We first  propose a quantum OV protocol in Section~\ref{sec3.1}, and then propose a quantum LCM computation protocol by taking the former as a
subprogram in Section~\ref{sec3.2}.

\subsection{Quantum one-vote-down vote protocol as a subprogram}\label{sec3.1}

We present our $OV$ protocol in \textbf{Procotol~2}, which is derived from Shi's SMQS protocol.

\begin{breakablealgorithm}
\caption{\raggedright\textbf{Protocol~2} Quantum one-vote-down vote (QOV)}
\begin{algorithmic}[1]
\renewcommand{\algorithmicrequire}{\textbf{Input}}
\renewcommand{\algorithmicensure}{\textbf{Output}}
\Require $n$ parties $P_0,P_1,\cdots,P_{n-1}$ have secret Boolean numbers $c_1,c_1,\cdots,c_{n-1} \in \set{0,1}$ respectively;
\Ensure Each $P_i$ gets $y=\prod_{k\in \intset{n}}{c_k}$, without any privacy of other parties;
\State Select a large enough positive integer $M$. Let $m=\floor{\log_2{(nM)}}+1$. $\forall i\in \intset{n}$, if $c_i=0$, then $P_i$ randomly selects a nonzero integer $x_i \in \set{1,2,\cdots,M}$, otherwise lets $x_i=0$;
\State {$P_0$ prepares two $m$-qubit particles $h,t$ initialized in $\bra{0}$. He applies $\QFT$ and $\CNOT^{\otimes m}$ just like in \textbf{Procotol~1}. Then he sends $t$ to $P_1$.}
\State {For $P_1$, he 
\begin{enumerate}[(1)]
\item {randomly selects an odd integer $q\in \intset{2^m}$, then applies $\U_{\times q}$ on $t$:
\begin{equation}
    \sfrac{2^m}\sum_{j\in \intset{2^m}}\bra{j}_h\bra{j}_t\onarrow{\U_{\times q}}\sfrac{2^m}\sum_{j\in \intset{2^m}}\bra{j}_h\bra{jq}_t;
\end{equation}}

\item {applies $\C_j$ on $t$:

\begin{equation}
    \sfrac{2^m}\sum_{j\in \intset{2^m}}\bra{j}_h\bra{jq}_t\bra{x_1}\onarrow{\C_j}\sfrac{2^m}\sum_{j\in \intset{2^m}}\pexp{x_1qj}{2^m}\bra{j}_h\bra{jq}_t\bra{x_1};
\end{equation}}

\item sends $t$ to $P_0$.
\end{enumerate}
}
\State {$P_0$ applies $\C_j$:
\begin{equation}
    \sfrac{2^m}\sum_{j\in \intset{2^m}}\pexp{x_1qj}{2^m}\bra{j}_h\bra{jp}_t\bra{x_0}\onarrow{\C_j}\sfrac{2^m}\sum_{j\in \intset{2^m}}\pexp{(x_1+x_0)qj}{2^m}\bra{j}_h\bra{jq}_t\bra{x_0};
\end{equation}
then if $n\ge 3$, sends $t$ to $P_2$; otherwise, sends $t$ to $P_1$, then the protocol jumps to Step 6.
}

\State {For $P_i,2\le i\le n-1 $, he 
\begin{enumerate}[(1)]
\item {applies $\C_j$:
\begin{equation}
\begin{aligned}
    &\sfrac{2^m}\sum_{j\in \intset{2^m}}\pexp{jq\sum_{k\in \intset{i}}x_k}{2^m}\bra{j}_h\bra{jp}_t\bra{x_i}\\
    &\onarrow{\C_j}\sfrac{2^m}\sum_{j\in \intset{2^m}}\pexp{jq\sum_{k\in \intset{i+1}}x_k}{2^m}\bra{j}_h\bra{jq}_t\bra{x_i};
 \end{aligned}
\end{equation}}
\item {then if $i\le n-2$, sends $t$ to $P_{i+1}$; otherwise, sends it to $P_1$.}
\end{enumerate}
}

\State {$P_1$ applies $\U_{\times q^{-1}}$ on $t$, where $q^{-1}$ is the multiplicative inverse of $q$ modulus $2^m$:
\begin{equation}
    \sfrac{2^m}\sum_{j\in \intset{2^m}}\pexp{jq\sum_{k\in \intset{n}}x_k}{2^m}\bra{j}_h\bra{jp}_t\onarrow{\C_j}\sfrac{2^m}\sum_{j\in \intset{2^m}}\pexp{jq\sum_{k\in \intset{n}}x_k}{2^m}\bra{j}_h\bra{j}_t;
\end{equation}
and sends $t$ to $P_0$.}

\State $P_0$ applies $\CNOT^{\otimes m}$ on $h,t$, measures $t$ to detect cheating, applies $\QFIT$ on $h$, and measures $h$ to get $z=q\sum_{k\in\intset{n}} {x_k} \mod 2^m$ at last; 
\State If $z>0$, $P_0$ lets $y=0$, otherwise $y=1$, then broadcasts $y$ to all other parties.
\end{algorithmic}
\end{breakablealgorithm}

\subsection{Secure multiparty quantum least common multiple computation protocol}\label{sec3.2}

We present our LCM protocol in \textbf{Procotol~3}.

\begin{breakablealgorithm}
\caption{\raggedright\textbf{Protocol~3} Secure multiparty quantum least common multiple computation (SMQLCMC)}
\begin{algorithmic}[1]
\renewcommand{\algorithmicrequire}{\textbf{Input}}
\renewcommand{\algorithmicensure}{\textbf{Output}}
\Require $n$ parties $P_0,P_1,\cdots,P_{n-1}$ have secret positive integers $x_1,x_1,\cdots$, $x_{n-1} \in \set{1,2,\cdots,2^m-1}$ respectively;
\Ensure Each $P_i$ gets $y=\setlcm{k\in \intset{n}}{x_k}$, without any privacy of other parties;
\State For $P_i,\forall i\in \intset{n}$, let $f_i:\intset{2^u}\rightarrow \intset{2^m}$ be $f_i(j)=j \mod{x_i}$.
\State {For $P_0$, he  
\begin{enumerate}[(1)]
    \item  prepares two $u=O(m)$-qubit quantum registers $h,t$ initialized as $\bra{0}_h\bra{0}_t$;
    \item {applies $\H^{\otimes u}$ on $h$:
        \begin{equation}
        \bra{0}_h\bra{0}_t\onarrow{\H^{\otimes u}}\sfrac{2^u}\sum_{j\in \intset{2^u}}\bra{j}_h\bra{0}_t;
        \end{equation}}
    \item {applies $\CNOT^{\otimes u}$ on $h,t$, where $h$ controls $t$:
        \begin{equation}
        \sfrac{2^u}\sum_{j\in \intset{2^u}}\bra{j}_h\bra{0}_t\onarrow{\CNOT^{\otimes u}}\sfrac{2^u}\sum_{j\in \intset{2^u}}\bra{j}_h\bra{j}_t;
        \end{equation}}
    \item prepares an $m$-qubit quantum registers $e_0$ initialized as $\bra{0}_{e_0}$;
    \item {applies $\U_{f_0}:\bra{j}_t\bra{0}_{e_0}\rightarrow \bra{j}_t\bra{f_0(j)}_{e_0}$ on $t,e_0$:
        \begin{equation}
        \sfrac{2^u}\sum_{j\in \intset{2^u}}\bra{j}_h\bra{j}_t\bra{0}_{e_0}\onarrow{\U_{f_0}}\sfrac{2^u}\sum_{j\in \intset{2^u}}\bra{j}_h\bra{j}_t\bra{f_0(j)}_{e_0};
        \end{equation}}
    \item sends $t$ to $P_1$;
\end{enumerate}}
\State {For $P_i, 1\le i\le n-1$, he
\begin{enumerate}[(1)]
    \item prepares an $m$-qubit quantum registers $e_i$ initialized as $\bra{0}_{e_i}$;
    \item {applies $\U_{f_i}:\bra{j}_t\bra{0}_{e_i}\rightarrow \bra{j}_t\bra{f_i(j)}_{e_i}$ on $t,e_i$:
    \begin{equation}
    \begin{aligned}
        &\sfrac{2^u}\sum_{j\in \intset{2^u}}\bra{j}_h\bra{j}_t\bra{f_0(j)}_{e_0}\bra{f_1(j)}_{e_1}\cdots\bra{f_{i-1}(j)}_{e_{i-1}}\bra{0}_{e_i}\\
        &\onarrow{\U_{f_0}}\sfrac{2^u}\sum_{j\in \intset{2^u}}\bra{j}_h\bra{j}_t\bra{f_0(j)}_{e_0}\bra{f_1(j)}_{e_1}\cdots\bra{f_{i-1}(j)}_{e_{i-1}}\bra{f_i(j)}_{e_i};
        \end{aligned}
    \end{equation}}
    \item sends $t$ to $P_{i+1\mod{n}}$;
\end{enumerate}}
\State {For $P_0$, he  
\begin{enumerate}[(1)]
    \item {applies $\CNOT^{\otimes u}$ on $h,t$, where $h$ controls $t$:
        \begin{equation}
        \sfrac{2^u}\sum_{j\in \intset{2^u}}\bra{j}_h\bra{j}_t\bra{f(j)}_{e}\onarrow{\CNOT^{\otimes u}}\sfrac{2^u}\sum_{j\in \intset{2^u}}\bra{j}_h\bra{0}_t\bra{f(j)}_{e},
        \end{equation}
        where $f(j)=f_0(j)\parallel f_1(j)\parallel \cdots \parallel f_{n-1}(j)$, $e=\left(e_0,e_1,\cdots,e_{n-1}\right)$;}
    \item measures $t$, if $t$ is not $\bra{0}$, then rejects, otherwise continues;
    \item {applies $\QFIT$ on $h$: 
        \begin{equation}
        \sfrac{2^u}\sum_{j\in \intset{2^u}}\bra{j}_h\bra{f(j)}_e \onarrow{\QFIT}\sfrac{T}\sum_{l\in \intset{T}}\bra{\phi}_h\bra{\widehat{f}(l)}_e,
        \end{equation}
        where $T$ is the period of $f(j)$ and $\phi\approx 2^u\frac{l}{T}$;}
    \item {measures $h$: 
        \begin{equation}
        \sfrac{T}\sum_{l\in \intset{T}}\bra{\phi}_h\bra{\widehat{f}(l)}_e \onarrow{Measure}\bra{\phi}_h\bra{\widehat{f}(l)}_e,
        \end{equation}
        where $l\in \intset{T}$ is selected with equal probability $\frac{1}{T}$;}
    \item uses continued fraction expansion of $\phi$ to get $\frac{l_1}{T_1}=\frac{l}{T}$, where $\frac{l_1}{T_1}$ is the minimalist fraction of $\frac{l}{T}$;
    \item tells others that the QPA process is completed;
\end{enumerate}}
\State $\forall i\in \intset{n}$, $P_i$ measures his register $e_i$, and tells $P_0$ he completes. 
\State $P_0$ broadcasts $T_1$ to all other parties. $\forall i\in \intset{n}$, if $x_i\lvert T$, then $P_i$ lets $c_i=1$, otherwise lets $c_i=0$; 
\State Implement \textbf{Protocol~2} to get $z=\prod_{k\in \intset{n}}{c_k}$, if $z=1$, the output $y=T_1$; otherwise, repeat the above steps until $z=1$.
\end{algorithmic}
\end{breakablealgorithm}

Since $f:\intset{2^u}\rightarrow \intset{2^{nm}}$ is an $nm$-bit function, its period $T<2^{nm}$, and we should let $u\ge 2nm+1$ to ensure QPA can output correct $\frac{l}{T}$. Similar to \textbf{Algorithm~1}, \textbf{Protocol~3} should be repeat $O\left( {\log{\log{T}} }\right)\le O\left(\log{nm}\right)$ times to get the correct $T$, since the probability that $l,T$ is coprime is $O\left(\frac{1}{\log{\log{T}} }\right)$. 

\section{Performance analysis}\label{sec4}

In this section, we analyze the performance of our protocol, including correctness, security and complexity.

\subsection{Correctness}\label{sec4.1}
We first prove that \textbf{Protocol~2} can indeed complete a one-vote-down vote. Since the core process of \textbf{Protocol~2} is consistent with \textbf{Protocol~1}, it's obvious that $P_0$ will get the correct $z=q\sum_{k\in\intset{n}}{x_k}\mod{2^m}$ in step 7. If $c_0=c_1=\cdots=c_{n-1}=1$, then  $x_0=x_1=\cdots=x_{n-1}=0$ in step 1, and consequently $z=0$, thus $y=1=\prod_{k\in\intset{n}}{c_k}$; If there is at least one $c_i=0$, then the corresponding $\sum_{k\in\intset{n}}{x_k}\mod{2^m}>0$ (because $m=2\ceil{\log_2{n}}\ge \log_2{n^2}$, i.e., $2^m\ge n^2\ge \sum_{k\in\intset{n}}{x_k}$). Since $q$ is an odd number, we have $z=q\sum_{k\in\intset{n}}{x_k}\mod{2^m}> 0$, thus $y=0=\prod_{k\in\intset{n}}{c_k}$. Therefore, \textbf{Protocol~2} correctly outputs $y=\prod_{k\in\intset{n}}{c_k}$, i.e., indeed completes a one-vote-down vote.

Then, we prove that \textbf{Protocol~3} can indeed output the LCM $y=\setlcm{k\in \intset{n}}{x_k}$. First, the correctness of the key QPA process (\textbf{Algorithm~1}) is proved\cite{1994Shor,1997Shor,2000Nielsen}, so we only need to prove that the period $T$ of the function $f$ is equal to $y$, as shown in \textbf{Theorem~\ref{theorem1}}.

\begin{theorem}\label{theorem1}
Assume we have $n$ function $f_0,f_1,\cdots,f_{n-1}$, where $\forall i\in \intset{n}$, $f_i:\intset{2^u}\rightarrow \intset{2^m}$ has a positive integer period $x_i<2^m$, and satisfies that for each pair of $j\ne j'\in \intset{2^u}$, $f_i(j)=f_i(j')$ only if $j\equiv j'(\mod x_i)$. Now let $f:\intset{2^u}\rightarrow \intset{2^v}$ ($v=nm$) satisfying $f(j)=f_0(j)\parallel f_1(j)\parallel \cdots \parallel f_{n-1}(j)$, then $f$ has a positive integer period $T=y=\setlcm{k\in \intset{n}}{x_k}\le 2^v$, and satisfies that for each pair of $j\ne j'\in \intset{2^u}$, $f(j)=f(j')$ only if $j\equiv j'(\mod T)$.
\end{theorem}

\begin{proof}
Suppose $j\ne j'\in \intset{2^u}$, if $f(j)=f(j')$, we have $f_0(j)\parallel f_1(j)\parallel \cdots \parallel f_{n-1}(j)= f_0(j')\parallel f_1(j')\parallel \cdots \parallel f_{n-1}(j')$, i.e., $f_0(j)=f_0(j'), f_1(j)=f_1(j'), \cdots$, $f_{n-1}(j)=f_{n-1}(j')$, thus we have $x_0\lvert(j-j') ,x_1\lvert(j-j'),\cdots,x_{n-1}\lvert(j-j')$, i.e., $y=\setlcm{k\in \intset{n}}{x_k}\lvert(j-j')$, then $j\equiv j'(\mod{y})$. Conversely, if $j\equiv j'(\mod{y})$, then $\forall i\in\intset{n}$, $x_i|(j-j')$, i.e., $j\equiv j'(\mod{x_i})$, thus $f_{i}(j)=f_{i}(j')$, then $f(j)=f(j')$. Therefore, for each pair of $j\ne j'\in \intset{2^u}$, $f(j)=f(j')$ only if $j\equiv j'(\mod y)$.
\end{proof}

Since $T=y=\setlcm{k\in \intset{n}}{x_k}\le 2^v$, we know \textbf{Protocol~3} will output the correct period $T$ of $f$, i.e., the result $y=\setlcm{k\in \intset{n}}{x_k}$ we need. If $T_1<T$, $Protocol~2$ in step 7 will detect the mistake. In summary, our LCM protocol is correct.

\subsection{Security}\label{sec4.2}
As a preliminary study, our security analysis is based on the semi honest model. In this case, all participants, as passive adversaries, will only measure or analyze the (classical or quantum) information in their own hands, and will not actively attack others or forge information. The opposite is the malicious model, where participants can actively steal others' information.

We first analyze the security of \textbf{Protocol~2}. Because all parties are semi-honest, they will correctly execute the protocol. Therefore, if anyone other than $P_0$ measures on $t$, he will get nothing, just like in \textbf{Protocol~2}. $P_0$ can only get $z=q\sum_{k\in\intset{n}}{x_k}\mod{2^m}$. Obviously, due to the existence of random odd number $q$, he cannot obtain the private information $x_i$ of a designated person $P_i$, and so as to $c_k$. However, we also need to analyze whether $z$ will disclose other information, i.e., the number of negative votes. We can prove that \textbf{Protocol~2} cannot hide the voting information with $100\%$ probability, but only with more than $1-1/M$ probability, as shown in \textbf{Theorem~\ref{theorem2}}.

\begin{theorem}\label{theorem2}
Only when $z=2^{m_1}s$, where $m_1>M$, \textbf{Protocol~2} does not disclose any information about the number of votes, and this probability $P_{vote}>1-\frac{1}{M}$.
\end{theorem}

\begin{proof}
Assume the number of negative votes is $\lambda$. Since $z=q\sum_{k\in\intset{n}}x_k\mod{2^m}$ is known, we set $Y=\sum_{k\in\intset{n}}x_k=2^{m_1}s$, where $s\in\intset{2^{m-m_1}}$ is an odd number and $m_1<m$ is a positive integer, then
\begin{equation}
    z=\left(sq2^{m_1}\mod{2^m}\right)=sq2^{m_1}-a2^m=2^{m_1}\left(sq+a2^{m-m_1}\right),
\end{equation}
Where $a$ is an integer and $sq+a2^{m-m_1}$ is an odd integer. In other words, $2^{m_1}$ can be deduced from $z$. Can $s$ be deduced too? For any $s'\in\intset{2^{m-m_1}}$, there is always $q'=s'^{-1}sq\mod{2^{m-m_1}}$ satisfying $s'q'\equiv sq(\mod{2^{m-m_1}})$, where $s'^{-1}$ is the modular $2^{m-m_1}$ multiplicative inverse of $s'$. Therefore, it is impossible to deduce $s$ from $z$.

Since we only can deduce $2^{m_1}$ from $z$, can we deduce $\lambda$? If each $P_i$ has selected the minimum value $x_i=1$, then $2^{m_1}\le \lambda \le n < 2^m$; If $P_i$ has selected the maximum value $x_i=M$, then $\frac{2^{m_1}}{M}\le \lambda \le \frac{n}{M}$. Therefore, we can know the range of $\lambda$ is $\frac{2^{m_1}}{M}\le \lambda \le n$. If $\frac{2^{m_1}}{M}\le 1$, i.e., $2^{m_1}\le M$, the range of $\lambda$ is $1\le \lambda \le n$, which gives no useful information. On the contrary, if $2^{m_1}> M$, it is possible to narrow the range. The question turns to whether $Y$ is a multiple of $2^{m_1}> M$. For the minimum $2^{m_1}> M$, it has about $\frac{\lambda M-\lambda}{2^{m_1}}$ multiples in $\lambda \sim \lambda M$, thus the probability of $2^{m_1}|Y$ is $P_{2^{m_1}|Y}\approx (\lambda M-\lambda)/\frac{\lambda M-\lambda}{2^{m_1}}=\frac{1}{2^{m_1}}<\frac{1}{M}$. Therefore, we deduce $\lambda$ only with probability $P_{vote}>1-\frac{1}{M}$.
\end{proof}

In other words, \textbf{Protocol~2}'s security and computing costs increase with the increase of $M$, so it is necessary to balance security and efficiency.

Now we analyze the security of \textbf{Protocol~3}. Considering that each $P_k$ has not sent the register $e_k$ carrying his own function $f_k$ to anyone else, there are only three known ways for each $P_i$ to obtain information:
\begin{enumerate}[(1)]
\item \textbf{Direct measurement attack} Before of after the QPA process is completed, $P_i$ can measure the register $h,t$ or $e_i$ they own to obtain any useful information;
\item \textbf{Pre-period-finding attack} Before the QPA process is completed, $P_i$ apply $\QFIT$ to his own register $h$ or $t$ to obtain the LCM of the parties who have completed their operations.
\item \textbf{Post-period-finding attack} $P_i$ copies a $\bra{j}$ (using $\CNOT$ gate) when he get the register $t$. After the QPA process is completed and before step 5, $P_i$ applies $\QFT$ on his copy to obtain any useful information (the reason we choose $\QFT$ is that $\bra{\widehat{f}(l)}=\sfrac{T}\sum_{k\in \intset{T}}\nexp{lk}{T}\bra{f(k)}$ is from $\bra{f(l)}$ by $\QFIT$ \cite{2000Nielsen}). 
\end{enumerate}

Accordingly, we give Theorem~\ref{theorem3}, which describes the security of \textbf{Protocol~3} under the above attacks.

\begin{theorem}\label{theorem3}
The above three attacks on \textbf{Protocol~3} will not give any useful information.
\end{theorem}

\begin{proof}
We discuss the transformation of quantum states under the above three attacks one by one.
\begin{enumerate}[(1)]
\item {\textbf{Direct measurement attack} Before the QPA process is completed, the state is
\begin{equation}\label{eq1}
\sfrac{2^u}\sum_{j\in \intset{2^u}}\bra{j}_h\bra{j}_t\bra{f_i(j)}_{e_i}\bra{f'(j)}_{e'},
\end{equation}
where $f'$ is the connected function of those who have completed their operations, and $e'$ is the register of $\bra{f'(j)}$. The attacker $P_i$ has register $e_i$, $h$ (if $i=0$) or $t$ (if $i\ne 0$). Assume he measure $h$ or $t$, then he will get a random $j\in\intset{2^u}$, which gives no useful information. If $P_i$ measures $e_i$, for arbitrary value $Y\in\intset{x_i}$ of function $f_i(j)=j\mod{x_i}$, its corresponding solution set is $f^{-1}_i(Y)=\set{kx_i+Y\lvert kx_i+Y\in\intset{2^u},k\in\Z}$, so $\left\lvert f^{-1}_i(Y)\right\lvert\approx \frac{2^u}{x_i}$. Therefore, for each $Y$, its probability is about $\frac{2^u}{x_i}\frac{1}{2^u}=\frac{1}{x_i}$. This probability distribution has nothing to do with the privacy of other parties, so it does not give any effective information, too. Consider that the attack occurs after the QPA process is completed, now the state is
\begin{equation}
\bra{\widehat{f}(l)}_e=\sfrac{T}\sum_{k\in \intset{T}}\nexp{lk}{T}\bra{f_i(k)}_{e_i}\bra{f'(k)}_{e'},
\end{equation}
which is similar to (\ref{eq1}). The derivation similar to the above shows that the measurement of ${e_i}$ still cannot give any useful information. In summary, direct measurement attack will not take effect.
}
\item {\textbf{Pre-period-finding attack} Similar to (1), the state is
\begin{equation}
\sfrac{2^u}\sum_{j\in \intset{2^u}}\bra{j}_h\bra{j}_t\bra{f'(j)}_{e'},
\end{equation}
where $f'$ is the connected function of those who have completed their operations (including $P_i$ or not), and $e'$ is the register of $\bra{f'(j)}$. Now $P_i$ will apply $\QFIT$ on $h$ or $T$ he owns (without losing generality, we choose $h$). The new state is 
\begin{equation}\label{eq2}
\begin{aligned}
&\sfrac{2^u}\sum_{j\in \intset{2^u}}\bra{j}_h\bra{j}_t\bra{f'(j)}_{e'}\\
&\onarrow{\QFIT}\sfrac{2^u}\sum_{j\in \intset{2^u}}\sfrac{2^u}\sum_{k\in \intset{2^u}}\nexp{jk}{2^u}\bra{k}_h\bra{j}_t\bra{f'(j)}_{e'}\\
&=\sfrac{2^u}\sum_{k\in \intset{2^u}}\bra{k}_h\sfrac{2^u}\sum_{j\in \intset{2^u}}\nexp{jk}{2^u}\bra{j}_t\bra{f'(j)}_{e'}.
\end{aligned}
\end{equation}
Since 
\begin{equation}
\begin{aligned}
&\sfrac{2^u}\sum_{j\in \intset{2^u}}\nexp{jk}{2^u}\bra{j}_t\bra{f'(j)}_{e'}\\
&=\U_{f'}\sfrac{2^u}\sum_{j\in \intset{2^u}}\nexp{jk}{2^u}\bra{j}_t\bra{0}_{e'}\\
&=\U_{f'}\QFIT\bra{k}_t\bra{0}_{e'}
\end{aligned}
\end{equation}
is a unit vector, we know that if $P_i$ measures on $h$, then $\forall k\in\intset{2^u}$, its probability is $\frac{1}{2^u}$, which gives no useful information. Consequently,
pre-period-finding attack will not take effect, too.
}
\item {\textbf{Post-period-finding attack} $P_i$ copies a $\bra{j}$ when he get the register $t$. Similar to (2), after step 4-(3), the state is
\begin{equation}
\begin{aligned}
\sfrac{2^u}\sum_{k\in \intset{2^u}}\bra{k}_h\sfrac{2^u}\sum_{j\in \intset{2^u}}\nexp{jk}{2^u}\bra{j}_t\bra{f(j)}_{e}.
\end{aligned}
\end{equation}
After $P_0$ measures on $h$, the state collapses to
\begin{equation}
\begin{aligned}
\bra{k}_h\sfrac{2^u}\sum_{j\in \intset{2^u}}\nexp{jk}{2^u}\bra{j}_t\bra{f(j)}_{e},
\end{aligned}
\end{equation}
where $k\in\intset{2^u}$ is random selected with probability $\frac{1}{2^u}$. Now $P_i$ applies $\QFT$ on $t$, then the new state is
\begin{equation}
\begin{aligned}
&\sfrac{2^u}\sum_{j\in \intset{2^u}}\nexp{jk}{2^u}\bra{j}_t\bra{f(j)}_{e}\\
&\onarrow{\QFT}\sfrac{2^u}\sum_{j\in \intset{2^u}}\nexp{jk}{2^u}\sfrac{2^u}\sum_{l\in \intset{2^u}}\pexp{jl}{2^u}\bra{l}_t\bra{f(j)}_{e}\\
&=\sfrac{2^u}\sum_{l\in \intset{2^u}}\bra{l}_t\sfrac{2^u}\sum_{j\in \intset{2^u}}\pexp{j(l-k)}{2^u}\bra{f(j)}_{e}\\
&=\sfrac{2^u}\sum_{l\in \intset{2^u}}\bra{l}_t\sfrac{2^u}\sum_{Y\in \intset{T}}\sum_{j\in f^{-1}(Y)}\pexp{j(l-k)}{2^u}\bra{Y}_{e}\\
&=\sfrac{2^u}\sum_{l\in \intset{2^u}}\bra{l}_t\bra{\psi_{l,k}}_e,
\end{aligned}
\end{equation}
where $T=\setlcm{k\in \intset{n}}{x_k}$ is the period of $f$, and for convenience, the vector $\bra{\psi_{l,k}}$ represents $\sfrac{2^u}\sum_{Y\in \intset{T}}\sum_{j\in f^{-1}(Y)}\pexp{j(l-k)}{2^u}\bra{Y}$. $\forall Y\in\intset{T}$, $f^{-1}(Y)=\set{qT+Y\lvert qT+Y\in\intset{2^u},q\in\Z}$, then $\left\lvert f^{-1}(Y)\right\lvert\approx \frac{2^u}{T}$. Therefore,
\begin{equation}
\begin{aligned}
&\left\lvert\sum_{j\in f^{-1}(Y)}\pexp{j(l-k)}{2^u}\right\lvert\approx\left\lvert\sum_{q=0}^{\frac{2^u}{T}-1}\pexp{(qT+Y)(l-k)}{2^u}\right\lvert^2\\
&=\left\lvert\pexp{Y(l-k)}{2^u} \right\lvert^2\left\lvert\sum_{q=0}^{\frac{2^u}{T}-1}\pexp{(qT)(l-k)}{2^u}\right\lvert^2=\left\lvert\sum_{q=0}^{\frac{2^u}{T}-1}\pexp{T(l-k)q}{2^u}\right\lvert^2.
\end{aligned}
\end{equation}
If $l-k\equiv r\frac{2^u}{T}\mod{2^u},r\in\Z$, then
\begin{equation}
\begin{aligned}
&\left\lvert\sum_{q=0}^{\frac{2^u}{T}-1}\pexp{T(l-k)q}{2^u}\right\lvert^2=\left\lvert\sum_{q=0}^{\frac{2^u}{T}-1}\pexp{Tr\frac{2^u}{T}q}{2^u}\right\lvert^2=\left\lvert\sum_{q=0}^{\frac{2^u}{T}-1}e^{\imath2\pi rq}\right\lvert^2\\
&=\left\lvert\sum_{q=0}^{\frac{2^u}{T}-1}{1}\right\lvert^2=\left(\frac{2^u}{T}\right)^2,
\end{aligned}
\end{equation}
thus $\left\langle \psi_{l,k}\lvert \psi_{l,k}\right\rangle=\frac{1}{2^u}\times T\times\frac{2^2u}{T^2}=\frac{2^u}{T}$. Then for each $l=k+ r\frac{2^u}{T}\mod{2^u},r\in\Z$, its probability is $\frac{2^u}{T}\frac{1}{2^u}=\frac{1}{T}$. Note that there are exactly $T$ $l=k+ r\frac{2^u}{T}\mod{2^u}$ ($r=0,1,\cdots,T-1$), thus the probability that $l=k+ r\frac{2^u}{T}\mod{2^u},r\in\Z$ is $T\times \frac{1}{T}=1$. Therefore, if $P_i$ measure $t$, he will only get a random $l=k+ r\frac{2^u}{T}\mod{2^u},r\in\intset{T}$ with probability $\frac{1}{T}$. However, because $k\in\intset{2^u}$ is also random selected known only by $P_0$, he can't get any useful information more than $T$ itself. Therefore, post-period-finding attack will also not take effect.
}
\end{enumerate}
In summary, the above three attacks on \textbf{Protocol~3} will not give any useful information.
\end{proof}

In addition to the above security, the existence of the \textbf{Protocol~2} in step 7 is an additional insurance, because once someone hinders the normal operation of the protocol, the output of the protocol will always be wrong. Then, because the protocol runs repeatedly too many times, the abnormal operation will be found.

It is worth noting that the security under direct measurement attack is the rationality of Step 5 of \textbf{Protocol~3}, because this step is actually equivalent to direct measurement attack, so no one's information will be disclosed. On the other hand, after destroying the quantum state in step 5, the used qubits can be recycled, thus saving resource consumption.  

In summary, \textbf{Protocol~3} is secure enough.

\subsection{Complexity}\label{sec4.3}

It's obvious that the time and communication complexity of \textbf{Protocol~2} are equal to \textbf{Protocol~1}, i.e., $O(nm^2)$ and $O(nm)$ respectively. Since $m=\floor{\log_2{(nM)}}+1=O\left(\log{n}\right)$ in \textbf{Protocol~2} ($M$ is a constant), they are $O(n\log^2{n})$ and $O(n\log{n})$ respectively.

Now we analyze the complexity of \textbf{Protocol~3}. In a repetition of \textbf{Protocol~3}, $\H$ and $\CNOT$ are $O(u)$,  $\QFIT$ is $O(u^2)$, and $\U_{f_i}:\bra{j}\bra{0}\rightarrow \bra{j}\bra{j \mod{x_i}}$ is $O(u^2)$ as a modular division operation. Therefore, one repetition costs $O(n^3 m^2+n\log^2{n})$ times fundamental quantum operations and $O(n^2m+n\log{n})$ times single-qubit communication, since $u=O(nm)$. Since we should repeat $O(\log{(nm)})$ times, the total time and communication complexity of \textbf{Protocol~3} are $O(n^3 m^2\log(nm)+n\log^2{n}\log(nm))$ and $O(n^2m\log(nm)+n\log{n}\log(nm))$ respectively. If $n\ll 2^m$, i.e., $\log n\ll m$, then the complexity become $O(n^3 m^2\log(nm))$ and $O(n^2m\log(nm))$ respectively.

\section{Conclusion}\label{sec5}

In this paper, based on the property that the period of a function connected from several periodic functions is the least common multiple (LCM) of all original periods, we propose a protocol for computing the LCM based on Shor's quantum period-finding algorithm. To verify the correctness of the results of our LCM protocol, we transform Shi's quantum summation protocol into a quantum one-vote-down veto protocol. We prove that our protocol can compute LCM with polynomial time and communication complexity, and enough security. The main shortcomings of our work are as follows: (1) The proposed protocol is probabilistic, and it must be carried out for many times cooperating with voting verification to calculate the correct result. Repetition itself brings some security problems; (2) Although the security of accurate information is solved, for statistical information, we can only ensure that security probability is large enough, not $100\%$; (3) We only analyzed the security under the semi honest model, and did not consider malicious attacks such as collusion attacks.

Our research reveals the extraordinary potential of quantum computing. Quantum computing is a future oriented computing theory, and many existing computer axioms may be subverted by quantum computing. What we need to do is not only to resist the attacks of quantum computing, but also to use it to achieve higher computing efficiency and stronger security than classical computing, thus laying a solid foundation for modern network information security.

\section*{Acknowledgements}
\noindent
This work is supported by the National Natural Science Foundation of China (62071240), the Innovation Program for Quantum Science and Technology (2021Z

\noindent D0302900), and the Research Innovation Program for College Graduates of Jiangsu Province, China (KYCX23\_1370).
%
%
%
\bibliographystyle{unsrt}

\end{document}